\documentclass[10pt,a4paper]{amsart}
\usepackage[margin=20mm]{geometry}
\usepackage[numbers]{natbib}
\usepackage{hyperref}

\usepackage{amssymb,amsmath}
\usepackage{amsthm}
\usepackage{bbm}
\usepackage{stmaryrd}
\usepackage[usenames,dvipsnames]{xcolor}
\usepackage{color}

\input xy
\xyoption{all} 

\numberwithin{equation}{section}

\newtheorem{theorem}{Theorem}[section]
\newtheorem{corollary}[theorem]{Corollary}

\newtheorem{proposition}[theorem]{Proposition}

\theoremstyle{definition}
\newtheorem{definition}[theorem]{Definition}
\newtheorem{remark}[theorem]{Remark}
\newtheorem{example}[theorem]{Example}

\newcommand{\p}{\mbox{\boldmath$\rho$}}
\newcommand{\Bs}{\mbox{\boldmath$s$}}

\newcommand{\rmd}{\textnormal{d}}

\newcommand{\rmh}{\textnormal{h}}

\DeclareMathOperator{\Vect}{Vect}

\DeclareMathOperator{\Ber}{Ber}

\font\black=cmbx10 \font\sblack=cmbx7 \font\ssblack=cmbx5 \font\blackital=cmmib10  \skewchar\blackital='177
\font\sblackital=cmmib7 \skewchar\sblackital='177 \font\ssblackital=cmmib5 \skewchar\ssblackital='177
\font\sanss=cmss10 \font\ssanss=cmss8 
\font\sssanss=cmss8 scaled 600 \font\blackboard=msbm10 \font\sblackboard=msbm7 \font\ssblackboard=msbm5
\font\caligr=eusm10 \font\scaligr=eusm7 \font\sscaligr=eusm5  \font\fraktur=eufm10
\font\sfraktur=eufm7 \font\ssfraktur=eufm5 
\font\bsymb=cmsy10 scaled\magstep2
\def\all#1{\setbox0=\hbox{\lower1.5pt\hbox{\bsymb
       \char"38}}\setbox1=\hbox{$_{#1}$} \box0\lower2pt\box1\;}
\def\exi#1{\setbox0=\hbox{\lower1.5pt\hbox{\bsymb \char"39}}
       \setbox1=\hbox{$_{#1}$} \box0\lower2pt\box1\;}

\def\tx#1{{\fam0\relax#1}}

\newfam\bifam
\textfont\bifam=\blackital \scriptfont\bifam=\sblackital \scriptscriptfont\bifam=\ssblackital

\newfam\blfam
\textfont\blfam=\black \scriptfont\blfam=\sblack \scriptscriptfont\blfam=\ssblack

\newfam\bbfam
\textfont\bbfam=\blackboard \scriptfont\bbfam=\sblackboard \scriptscriptfont\bbfam=\ssblackboard

\newfam\ssfam
\textfont\ssfam=\sanss \scriptfont\ssfam=\ssanss \scriptscriptfont\ssfam=\sssanss
\def\sss#1{{\fam\ssfam\relax#1}}

\newfam\clfam
\textfont\clfam=\caligr \scriptfont\clfam=\scaligr \scriptscriptfont\clfam=\sscaligr

\newfam\frfam
\textfont\frfam=\fraktur \scriptfont\frfam=\sfraktur \scriptscriptfont\frfam=\ssfraktur

\def\hpb#1{\setbox0=\hbox{${#1}$}
    \copy0 \kern-\wd0 \kern.2pt \box0}
\def\vpb#1{\setbox0=\hbox{${#1}$}
    \copy0 \kern-\wd0 \raise.08pt \box0}

\def\pmb#1{\setbox0\hbox{${#1}$} \copy0 \kern-\wd0 \kern.2pt \box0}
\def\pmbb#1{\setbox0\hbox{${#1}$} \copy0 \kern-\wd0
      \kern.2pt \copy0 \kern-\wd0 \kern.2pt \box0}
\def\pmbbb#1{\setbox0\hbox{${#1}$} \copy0 \kern-\wd0
      \kern.2pt \copy0 \kern-\wd0 \kern.2pt
    \copy0 \kern-\wd0 \kern.2pt \box0}
\def\pmxb#1{\setbox0\hbox{${#1}$} \copy0 \kern-\wd0
      \kern.2pt \copy0 \kern-\wd0 \kern.2pt
      \copy0 \kern-\wd0 \kern.2pt \copy0 \kern-\wd0 \kern.2pt \box0}
\def\pmxbb#1{\setbox0\hbox{${#1}$} \copy0 \kern-\wd0 \kern.2pt
      \copy0 \kern-\wd0 \kern.2pt
      \copy0 \kern-\wd0 \kern.2pt \copy0 \kern-\wd0 \kern.2pt
      \copy0 \kern-\wd0 \kern.2pt \box0}

\mathchardef\za="710B  
\mathchardef\zb="710C  
\mathchardef\zg="710D  
\mathchardef\zd="710E  
\mathchardef\zve="710F 
\mathchardef\zz="7110  
\mathchardef\zh="7111  
\mathchardef\zvy="7112 
\mathchardef\zi="7113  
\mathchardef\zk="7114  
\mathchardef\zl="7115  
\mathchardef\zm="7116  
\mathchardef\zn="7117  
\mathchardef\zx="7118  
\mathchardef\zp="7119  
\mathchardef\zr="711A  
\mathchardef\zs="711B  
\mathchardef\zt="711C  
\mathchardef\zu="711D  
\mathchardef\zvf="711E 
\mathchardef\zq="711F  
\mathchardef\zc="7120  
\mathchardef\zw="7121  
\mathchardef\ze="7122  
\mathchardef\zy="7123  
\mathchardef\zf="7124  
\mathchardef\zvr="7125 
\mathchardef\zvs="7126 
\mathchardef\zf="7127  
\mathchardef\zG="7000  
\mathchardef\zD="7001  
\mathchardef\zY="7002  
\mathchardef\zL="7003  
\mathchardef\zX="7004  
\mathchardef\zP="7005  
\mathchardef\zS="7006  
\mathchardef\zU="7007  
\mathchardef\zF="7008  
\mathchardef\zW="700A  
\mathchardef\zC="7009  

\newcommand{\be}{\begin{equation}}
\newcommand{\ee}{\end{equation}}

\newcommand{\bea}{\begin{eqnarray}}
\newcommand{\eea}{\end{eqnarray}}
\def\*{{\textstyle *}}
\newcommand{\R}{{\mathbb R}}

\newcommand{\Z}{{\mathbb Z}}

\newcommand{\s}{{\textstyle *}}

\def\Vect{\sss{Vect}}

\def\sT{{\sss T}}

\def\xi{\tx{i}}

\def\s*{{\scriptstyle *}}

\def\cO{\mathcal{O}}

\newcommand{\beas}{\begin{eqnarray*}}
\newcommand{\eeas}{\end{eqnarray*}}

\def\half{\frac{1}{2}}

\title{Modular Classes of Q-Manifolds, Part II:\\   Riemannian Structures \& Odd Killing Vectors Fields }

   \author{Andrew James Bruce} 
   \address{Mathematics Research Unit, University of Luxembourg, Maison du Nombre 6, avenue de la Fonte, 
L-4364 Esch-sur-Alzette}  
   \email{andrewjamesbruce@googlemail.com}
  
\date{\today}
 
\begin{document}

\begin{abstract}
We define and make an initial study of  (even) Riemannian supermanifolds equipped with a homological vector field that is also a Killing vector field. We refer to such supermanifolds as \emph{Riemannian Q-manifolds}. We show that such Q-manifolds are unimodular, i.e., come equipped with a Q-invariant Berezin volume.  \par
\end{abstract}

 \maketitle

\setcounter{tocdepth}{2}
 \tableofcontents

\section{Introduction} 
This paper is a direct continuation of an earlier paper by the author \cite{Bruce:2017} in which the notion of the modular class of a Q-manifold was reviewed and various illustrative examples are given. Q-manifolds (see \cite{Schwarz:1993}), i.e., supermanifolds equipped with an odd vector field that `squares to zero', have become an important part of mathematical physics due to their prominence in the AKSZ-formalism \cite{Alexandrov:1997} and the conceptionally neat formalism they provide for describe Lie algebroids \cite{Vaintrob:1997} and Courant algebroids \cite{Roytenberg:2002}, as well as various generalisations thereof. The modular class of a Q-manifold (see \cite{Lyakhovich:2004,Lyakhovich:2010}) is a natural generalisation of the modular class of a Lie algebroid  \cite{Evans:1999}. \par 
 The modular class of a Q-manifold is given in terms of the divergence of the homological vector field, though it does not depend on the chosen Berezin volume. The vanishing of the modular class is a necessary and sufficiency condition for the existence of a $Q$-invariant Berezin volume. Q-manifolds with vanishing modular class are known as \emph{unimodular Q-manifolds}.  Here we given another class of examples of unimodular Q-manifolds by considering (even) Riemannian supermanifolds that admit an odd Killing vector field that is homological.  We will refer to such supermanifolds as \emph{Riemannian Q-manifolds}. To our knowledge, such supermanifolds have not appeared in the literature before now. The notion of \emph{supersymmetric Killing structures} appears in the work of Klinker \cite{Klinker:2005}.\par 
Riemannian Q-manifolds are reminiscent of even symplectic supermanifolds in the sense that Killing vector fields are akin to Hamiltonian vector fields. Moreover, we have a version of Liouville's theorem on even symplectic supermanifolds that states that there is always a Berezin volume that is invariant with respect to all Hamiltonian vector fields. This implies, for example, that the modular class of a Courant algebroid (or more properly, a symplectic Lie 2-algebroid \cite{Roytenberg:2002}) vanishes.  The direct analogue of this is explicitly proved in this paper, though the result should not come as a surprise: the canonical Berezin volume on a Riemannian supermanifold is invariant under the action of Killing vector fields. This directly implies that the modular class of a Riemannian Q-manifold vanishes.  This paper is devoted to  explicitly proving this. Moreover, at each stage, we give concrete examples. \par 
 An incomplete list of relatively recent papers on Riemannian supermanifolds includes \cite{Galaev:2012,Garnier:2014,Garnier:2012,Goertsches:2008,Groeger:2014,Klaus:2019}. We do not believe that this paper contains anything truly new about Riemannian supergeometry. However, finding clear references to the expressions we require is not so easy. Thus, part of this paper is devoted to setting-up what we need to describe Riemannian Q-manifolds.
\medskip

\noindent \textbf{Arrangement.}  In Section \ref{sec:RiemSup} we recall the basic facets of Riemannian supergeometry relevant to our needs. In particular, we pay attention to  Killing vector fields, the canonical Berezin volume and the divergence operator. We then move on to Q-manifolds and their modular classes in Section \ref{sec:QMod}. Much of this section is taken from \cite{Bruce:2017} and references therein. In Section \ref{sec:RiemQ} we define the notion of a Riemannian Q-manifold and explore some of their basic properties.  We end  with Section \ref{sec:Conc} with a few concluding remarks.

\medskip
\noindent \textbf{Our use of supermanifolds.} We assume that the reader has some familiarity with the basics of the theory of supermanifolds. We will  understand a \emph{supermanifold} $M := (|M|, \:  \cO_{M})$ of dimension $n|m$  to be a supermanifold in the sense of Berezin \& Leites  \cite{Berezin:1976}, i.e., as a locally superringed space that is locally isomorphic to $\mathbb{R}^{n|m} := \big (\R^{n}, C^{\infty}(\R^{n})\otimes \Lambda(\zx^{1}, \cdots \zx^{m}) \big)$. In particular,  given any point on $|M|$ we can always find a `small enough' open neighbourhood $|U|\subseteq |M|$ such that we can  employ local coordinates $x^{a} := (x^{\mu} , \zx^{i})$ on $M$.  We will call (global) sections of the structure sheaf \emph{functions}, and often  denote the supercommutative algebra of all functions as $C^{\infty}(M)$. The underlying smooth manifold $|M|$ we refer to as the \emph{reduced manifold}.  We will make heavy use of local coordinates on supermanifolds and employ the standard abuses of notation when it comes to describing, for example, morphisms of supermanifolds. We will denote the Grassmann parity of an object $A$ by `tilde', i.e., $\widetilde{A} \in \Z_{2}$.  By `even' and `odd' we will be referring to the Grassmann parity of  the objects in question. As we will work in the category of smooth supermanifolds, all the algebras, commutators etc. will be $\Z_{2}$-graded.\par
The \emph{tangent sheaf} $\mathcal{T}M$ of a supermanifold $M$ is the sheaf of derivations of sections of the structure sheaf -- this is, of course, a sheaf of locally free $\cO_{M}$-modules. Global sections of the tangent sheaf we refer to as \emph{vector fields}, and denote the $\cO_{M}(|M|)$-module of vector fields as $\Vect(M)$. The total space of the tangent sheaf we will denote by $\sT M$ and refer to this as the \emph{tangent bundle}. By shifting the parity of the fibre coordinates one obtains the \emph{antitangent bundle} $\Pi \sT M$. We will reserve the nomenclature \emph{vector bundle} for the total space of a sheaf of locally free $\cO_{M}$-modules, that is we will be referring to `geometric vector bundles'.\par
There are several good books on the subject of supermanifolds and we suggest  Carmeli, Caston \& Fioresi \cite{Carmeli:2011},  Manin \cite{Manin:1997}  and Varadrajan \cite{Varadrajan:2004} as general references. The encyclopedia edited by Duplij, Siegel \& Bagger  \cite{Duplij:2004} is also indispensable, as is the review paper by Leites \cite{Leites:1980}.  DeWitt \cite[Section 2.8]{DeWitt:1992} discusses in some detail Riemannian geometry on DeWitt--Rogers supermanifolds. While some care is needed in translating between supermanifolds (as locally ringed spaces) and DeWitt--Rogers supermanifolds, most of the expressions given by DeWitt on Riemannian structures remain valid in Riemannian supergeometry.

\section{Riemannian supermanifolds}\label{sec:RiemSup}
\subsection{The tangent bundle of a supermanifold and symmetric tensors}
The \emph{tangent bundle} $\sT M$ of a supermanifold $M$, we define as a natural bundle via local coordinates in almost exactly the same way as one can for a smooth manifold. For convenience, we sketch the construction here.  \par  
Let $M = (|M|, \cO_M)$ be a supermanifold  equipped with an atlas $\{ U_i, \rmh_i\}_{i \in \mathcal{I}}$. Here $|U_i| \subset |M|$ form an open cover of $M$ and $U_i =  (|U_i| , \cO_M|_{|U_i|})$. The maps
$$\rmh_i :  U_i \longrightarrow \mathcal{U}_i^{n|m}$$
are supermanifold diffeomorphisms. Here  $\mathcal{U}_i^{n|m}$ are superdomains, i.e., open subsupermanifolds of $\R^{n|m}$.  Over non-empty $|U_{ij}| = |U_i| \cap |U_j|$ we have transition functions (induced glueing data)
$$\rmh_j \circ \rmh^{-1}_i : ~  \mathcal{U}_i^{n|m} \longrightarrow \mathcal{U}_j^{n|m}\,,$$
were we have neglected to write out the obvious restrictions. It is clear that such maps satisfy the cocycle conditions and so constitute glueing data. Suppose that we have coordinates $x^{a'}$ on $\mathcal{U}_j^{n|m}$ and $x^{a}$ on $\mathcal{U}_i^{n|m}$. Then the changes of coordinates we write as
$$x^{a'} = x^{a'}(x),$$
by employing the standard abuses of notation. \par 
We define the tangent bundle $\sT M$ by its atlas $\{\sT U_i, \sT \rmh_i \}_{i \in \mathcal{I}}$ induced from the given atlas on $M$. That is, given any $U_i$ in the atlas we have
$$\sT \rmh_i :  ~ \sT U_i \longrightarrow  \mathcal{U}_i^{n|m} \times \R^{n|m}\,.$$
Clearly, $|\sT U_i| \cong \mathcal{U}^n_i \times \R^n $.  The induced glueing data is easiest to explain using natural coordinates $(x^a , \dot{x}^b)$. Again using the standard abuses of notation, the admissible coordinate transformations are of the form
\begin{align*}
x^{a'}= x^{a'}(x), && \dot{x}^{b'} = \dot{x}^{b}\left(\frac{\partial x^{b'}}{\partial x^b}\right)\,.
\end{align*}
One can show that we do indeed construct a supermanifold of dimension $2n|2m$ in this way. Moreover, it is clear that we have a vector bundle structure on $\sT M$. As such, the tangent bundle can be considered as a non-negatively graded supermanifold (see \cite{Grabowski:2012, Roytenberg:2002, Voronov:2002}).  In particular, we assign weight zero to the base coordinates $x$ and weight one to the fibre coordinates $\dot{x}$. As the admissible coordinate transformations respect the assignment of weight,  it makes sense to speak of functions on $\sT M$ of a given weight. Moreover, it is known that homogeneous functions on $\sT M$ are monomial on the fibre coordinates.  We will denote the \emph{polynomial algebra} on $\sT M$ as $\mathcal{A}(\sT M)$. Clearly, $\mathcal{A}^0(\sT M) = C^\infty(M)$. Note that the polynomial algebra as a natural (right) $C^\infty(M)$-module structure.  We will denote the submodule of monomials of degree $k$ as $\mathcal{A}^k(\sT M)$. We  make the following definition.
\begin{definition}
The $C^\infty(M)$-module of rank $k$ \emph{symmetric covariant tensors} on a supermanifold $M$ is defined to be the $C^\infty(M)$-module of monomials on $\sT M$ of weight $k$.
\end{definition}
Locally in natural coordinates, $T \in \mathcal{A}^k(\sT M)$ looks like
$$T = \dot{x}^{a_1} \dot{x}^{a_2} \cdots \dot{x}^{a_k} T_{a_k \cdots a_2 a_1}(x)\,$$
where the components $T_{a_k \cdots a_2 a_1}$ are (super)symmetric.

\subsection{Riemannian structures}
\begin{definition}
A  \emph{Riemannian metric} on a supermanifold $M$, is an even, symmetric, non-degenerate, $\cO_M$-linear morphisms of sheaves
$$\mathcal{T}M \otimes_{\cO_M} \mathcal{T}M \longrightarrow \cO_M.$$
A \emph{Riemannian supermanifold} is a supermanifold equipped with a Riemannian metric.
\end{definition}
In terms of vector fields, we have the following properties:
\begin{enumerate}
\item$\widetilde{\langle X| Y \rangle_g} = \widetilde{X} + \widetilde{Y}$;
\item $\langle X| Y \rangle_g = (-1)^{\widetilde{X} \, \widetilde{Y}} \langle Y| X \rangle_g$;
\item If $\langle X| Y \rangle_g =0$  for all $Y \in \Vect(M)$, then $X =0$;
\item $\langle f X + Y| Z \rangle_g = f \langle X|Z \rangle_g + \langle Y| Z\rangle_g$, 
\end{enumerate} 
For all (homogeneous) $X,Y,Z \in \Vect(M)$ and $f \in C^\infty(M)$.   
\begin{remark}
A Riemannian metric on $M$ naturally induces a pseudo-Riemannian metric on the reduced manifold $|M|$.  As we will not explicitly make use of this reduced structure we will not spell-out the construction. 
\end{remark}
A Riemannian metric is specified by an even degree two function $g \in \mathcal{A}^2(\sT M)$, i.e., a Grassmann degree zero rank $2$ symmetric covariant tensor.  In local coordinates, we write
$$g(x, \dot{x}) = \dot{x}^a \dot{x}^b \,g_{ba}(x).$$  
Under changes of coordinates $x^a \mapsto x^{a'}(x)$ the components of the metric transform as 
$$g_{b'a'}(x') = (-1)^{\widetilde{a}' \, \widetilde{b}} \left(\frac{\partial x^b}{\partial x^{b'}}\right) \left(\frac{\partial x^a}{\partial x^{a'}}\right)~  g_{ab}\,,$$
where we have explicitly used the symmetry $g_{ab} = (-1)^{\widetilde{a} \, \widetilde{b}} \, g_{ba}$.\par 
If we denote the vertical lift of a vector field by $\iota_X$, which in local coordinates is given by
$$X = X^a(x)\frac{\partial}{\partial x^a} \rightsquigarrow \iota_X :=  X^{a}(x)\frac{\partial}{\partial \dot{x}^a} \in \Vect(\sT M)\,,$$
then we observe that
$$\langle X| Y \rangle_g = \frac{1}{2} \iota_X \iota_Y g\,,$$ 
which leads to the local expression 
$$  \langle X| Y \rangle_g =  (-1)^{\widetilde{Y}\, \widetilde{a}} ~ X^a(x)Y^b(x) g_{ba}(x).$$
It is a  straightforward exercise to show that the above local expression for the metric pairing is invariant under changes of coordinates. \par 
 It is well-known that the non-degeneracy condition forces the dimensions of the supermanifold $M$ to be $n | 2\,p$, i.e., we require an even number of odd dimensions. 
\begin{example}
As any manifold can be considered as a supermanifold with vanishing `odd directions', i.e., a supermanifold of dimension $n|0$, \emph{any} (pseudo-)Riemannian manifold can be considered as a Riemannian supermanifold. 
\end{example}
\begin{example}
Consider $\R^{1|2}$ equipped with canonical global coordinates $(t, \zx^1, \zx^2)$.  Any vector field decomposes as
$$X = X^0 \frac{\partial}{\partial t} + X^1 \frac{\partial}{\partial \zx^1} + X^2 \frac{\partial}{\partial \zx^2}\,, $$
where each component is a function of the canonical coordinates. The standard metric is given by
$$g = (\dot{t})^2 \pm 2 \,\dot{\zx}^1 \dot{\zx}^2\,,$$
where we have a choice with the sign for the `odd part' of the metric. Then a simple calculation gives
$$\langle X, | Y\rangle_g =  X^0 Y^0 \pm (-1)^{\widetilde{Y}} (X^1 Y^2 - X^2 Y^1).$$
\end{example}
\begin{example}
Consider $\R^{3|2}$ equipped with standard global coordinates $(x,y,z, \zx^1, \zx^2)$. The equation
$$x^2 + y^2 + z^2 - 2 \, \zx^1 \zx^2 =1 $$
defines the super-sphere $\mathbb{S}^{2|2} \subset \R^{3|2}$ (using slight abuse of notation). As  (local)  coordinates on $\mathbb{S}^{2|2}$ we can use  the  standard angles $(\theta, \phi)$,  i.e.,   the  coordinates inherited  from using  polar  coordinates on $\R^3$,  complemented  by  $(\zx^1, \zx^2)$  inherited  from  the `super-environment'.  The reduced  manifold is  standard two-sphere.  As a sub-supermanifold of the Riemannian supermanifold $\R^{3|2}$, the super-sphere is equipped with a non–degenerate metric inhered from the embedding.  This metric is given by
$$g = \dot{\theta}^2 + \sin^2 \theta \, \dot{\phi} - 2 \,\dot{\zx}^1 \dot{\zx}^2 \,.$$
\end{example}
\begin{example}
Let $M$ be an almost symplectic manifold, i.e., a manifold equipped with a non-degenerate two-form $\omega$, that this not necessary closed. This forces the dimension of $M$ to be even. Furthermore, let us assume that $M$ is equipped with a Riemannian metric, which we will denote as $h$. It is always possible to equip \emph{any} smooth manifold with a Riemannian metric and we will not require any compatibility condition between the almost symplectic structure $\omega$ and the Riemannian structure $h$. We want to build a Riemannian metric on the supermanifold $\Pi \sT M$.  To do this, consider the double supervector bundle $\sT(\Pi \sT M)$, which we equip with natural coordinates  $(x^a, \rmd x^b, \dot{x}^c , \rmd \dot{x}^d)$. Admissible changes of coordinates are of the form (using standard abuses of notation)
\begin{align*}
&x^{a'} = x^{a'}(x), & \rmd x^{b'} = \rmd x^a \frac{\partial x^{b'}}{\partial x^a},\\
&\dot{x}^{c'} = \dot{x}^b \frac{\partial x^{c'}}{\partial x^b},& \rmd\dot{x}^{d'} = \rmd\dot{x}^{c}  \frac{\partial x^{b'}}{\partial x^c} + \dot{x}^{b} \rmd{x}^c \frac{\partial^2 x^{d'}}{ \partial x^c \partial x^b}.
\end{align*}
The Levi-Civita connection $\nabla$ associated with the metric induces a splitting 
$$\sT (\Pi \sT M) \stackrel{\phi_h}{\xrightarrow{\hspace*{25pt}}} \Pi \sT M \times_M \sT M \times_M \Pi \sT M\,,$$
which we write in natural coordinates as
$$\phi^*_h \zx^a = \rmd\dot{x}^a + \rmd{x}^b \dot{x}^c \Gamma^a_{cb}(x) =: \nabla \dot{x}^a.$$
Here $\zx^a$ are the (fibre)  coordinates on last factor of the decomposed or split double supervector bundle.  The splitting $\phi_h$ is understood as acting as the identity on the remaining coordinates, i.e., we just canonically make the required identifications. On the decomposed  double supervector bundle we can take the sum of the Riemannian metric and the almost symplectic structure. In natural coordinates we have
$$G :=  \dot{x}^a \dot{x}^b g_{ba}(x) + \zx^a \zx^b \omega_{ba}(x).$$
The metric on $\sT(\Pi \sT M)$ is then the pull-back of $G$ by the splitting.  Thus, we write
$$g = \phi^*_h G  = \dot{x}^a \dot{x}^b g_{ba}(x) + \nabla \dot{x}^a \nabla \dot{x}^b \omega_{ba}(x).$$
 
\end{example}

\begin{remark}
Odd Riemannian structures can similarly be defined. There are no changes to the above definition except that the parity now is shifted, i.e., the pairing between two vector fields will now be $\widetilde{X} + \widetilde{Y} +1$. The condition of being non-degenerate now forces there to be an equal number of even and odd dimensions.  We will only consider even metrics in this paper. The reason, in part, is that while even metrics, together with even and odd symplectic structures, have found application in physics, odd Riemannian structures remain a mathematical curiosity. 
\end{remark}
All the standard constructions of classical Riemannian geometry generalise to Riemannian supermanifolds, for example the fundamental theorem holds. We will not make use of the Levi-Civita connection or the curvature tensors in this paper. They can all be defined via minor sign modifications of the classical definitions (see for example \cite{Monterde:1996}).
\begin{remark}
There is also the notion of a quasi-Riemannian structure due to Mosman \& Sharapov \cite{Mosman:2011}, which intriguingly exists on any supermanifold.  This structure understood as a pair $(G,\nabla)$, where $G$ symmetric positive definite tensor field of type $(0,2)$ and  $\nabla$ is a compatible affine connection, which in general is not symmetric. Naturally, an even Riemannian structures and metric compatible, but not necessarily torsion free affine connection is an example of a quasi-Riemannian structure. 
\end{remark}

\subsection{Killing vector fields}
Killing vector fields are defined in exactly the same way as in classical Riemannian geometry.
\begin{definition}\label{def:KillVF}
A vector field $X \in \Vect(M)$ is said to be a Killing vector field if and only if
$$L_X g =0\,.$$
\end{definition}
At this juncture, we need to explain the above Lie derivative and derive a local expression. Recall that any homogeneous vector field $X \in \Vect(M)$ defines a  local infinitesimal diffeomorphism (see \cite[\S 2.3.9.]{Voronov:2014}) of $\sT M$, which in local coordinates is of the form
\begin{align*}
& x^a \mapsto x^a + \lambda \, X^a(x)\,,\\
& \dot{x}^a \mapsto  \dot{x}^a + \lambda \, \dot{x}^b \frac{\partial X^a}{\partial x^b \hfill}(x)\,,
\end{align*} 
where $\lambda$ is an external parameter of degree $\widetilde{\lambda} = \widetilde{X}$. Under this local diffeomorphism  a quick calculation shows that the metric  $g$ changes as 
$$g(x, \dot x) \mapsto g(x, \dot x) + \lambda \, \dot{x}^a \dot{x}^b \left((-1)^{\widetilde{X} \, \widetilde{a}} \frac{\partial X^c}{\partial x^b} g_{ca} + (-1)^{\widetilde{b}(\widetilde{X} + \widetilde{a})}  \frac{\partial X^c}{\partial x^a} g_{cb} + (-1)^{\widetilde{X}(\widetilde{a} + \widetilde{b})} X^c\frac{\partial g_{ba}}{\partial x^c}  \right) + \cO(\lambda^2).$$
By definition, locally, the Lie derivative is given by the first-order term in $\lambda$. Thus, we have the local expression
\begin{equation}\label{eqnm:LieDerMet}
(L_Xg)_{ba} = (-1)^{\widetilde{X} \, \widetilde{a}} \frac{\partial X^c}{\partial x^b} g_{ca} + (-1)^{\widetilde{b}(\widetilde{X} + \widetilde{a})}  \frac{\partial X^c}{\partial x^a} g_{cb} + (-1)^{\widetilde{X}(\widetilde{a} + \widetilde{b})} X^c\frac{\partial g_{ba}}{\partial x^c}.
\end{equation}
Naturally, this local expression is identical to the classical one up to some sign factors. 
\begin{proposition}
 The set of all Killing vector fields on even Riemannian supermanifold $(M,g)$ forms a Lie algebra with respect to the standard Lie bracket of vector fields on $M$.
\end{proposition}
\begin{proof}
This follows in complete parallel with the classical case  using $L_{[X,Y]} = [L_X, L_Y]$.
\end{proof}

\subsection{The inverse metric and the trace}
The non-degeneracy of a metric implies that the components, thought of as a rank-2 covariant tensor, is invertible.  The defining relation for the \emph{inverse metric} is
$$g^{ac}g_{cb} = g_{bc}g^{ca} = \delta_b^a\,,$$
just as it is on a classical Riemannian manifold.  Clearly, the inverse metric is even. The above relation allows us  to deduce the symmetry property of the inverse metric.
\begin{proposition}
The inverse metric $g^{ab}$ has the following symmetry:
$$(-1)^{\widetilde{b}} \, g^{ab} = (-1)^{\widetilde{a} \, \widetilde{b} + \widetilde{a}} \, g^{ba}\,.$$
\end{proposition}
\begin{proof}
Let $g^{ab} = (-1)^\lambda \, g^{ba} $, where $\lambda$ is to be determined. From the defining relation and the symmetry of the metric we have
$$g^{ac}g_{cb} = (-1)^{\widetilde{a} \widetilde{c} + \widetilde{a} \widetilde{b} + \widetilde{c} + \lambda} \, g_{bc}g^{ca}\, .$$
Then, once $a = b$ we see that $\lambda = \widetilde{a} \, \widetilde{c} + \widetilde{a} + \widetilde{c}$. This gives the required symmetry. 
\end{proof}
\begin{definition}
Let $(M,g)$ be a Riemannian supermanifold  we define the \emph{metric trace} or just \emph{trace} as the $C^\infty(M)$-linear map
$$\mathcal{A}^2(\sT M) \longrightarrow C^\infty(M)\,,$$
given in local coordinates as
$$\textnormal{Str}_g T := (-1)^{\widetilde{a}} \, g^{ab}T_{ba}\,,$$
for any arbitrary $T = \dot{x}^a\dot{x}^b T_{ba}(x) \in \mathcal{A}^2(\sT M)$. 
\end{definition}
In words, the metric trace is given by contraction of the rank two symmetric rank two covariant tensor with the inverse metric to form a matrix, and then we take the standard supertrace. 
\begin{remark}
The metric trace can also be defined for rank two covariant tensors without any symmetry condition. We focus on the symmetric case as this is what we will need in later sections of this paper. 
\end{remark}

\subsection{The divergence operator and  the canonical Berezin volume}
Let us for simplicity assume that the supermanifolds that we will be dealing with are superoriented (see \cite{Shander:1988} and/or \cite[page 285]{Duplij:2004}). That is the underlying reduced manifold will be oriented, and we further require that we have chosen an atlas such that the Jacobian associated with any change of coordinates is strictly positive. The \emph{Berezin bundle} $\Ber(M)$, is understood as the (even) line bundle over $M$ whose sections in a local trivialisation are of the form
$$ \Bs= D[x] s(x),$$ 
where $D[x]$ is the \emph{coordinate volume element}. Under changes of local coordinate  we have
$$D[x'] = D[x] \Ber\left(\frac{\partial x'}{\partial x}  \right).$$
Sections of $\Ber(M)$ are \emph{Berezin forms} on $M$. Note the the Grassmann parity of a Berezin density is determined by $s(x)$. A \emph{Berezin volume} on $M$ is a nowhere vanishing  even Berezin form. \par
In the classical case on a manifold, one needs a volume form (or in the non-oriented case a density) in order to define the divergence of a vector field. The same is true for supermanifolds, and we take the definition  of the \emph{divergence of a vector field} $X \in \Vect(M)$ with respect to a chosen Berezin volume to be
\begin{equation} \label{eqn:DivVect}
\p \: \textnormal{Div}_{\p}X = L_{X}\p\,.
\end{equation}
In local coordinates, this definition amounts to
\begin{equation}\label{eqn:DivLC}
\textnormal{Div}_{\p}X = (-1)^{\widetilde{a}(\widetilde{X}+1)} \: \frac{1}{\rho} \frac{\partial}{\partial x^{a}}\left(X^{a} \rho \right).
\end{equation}
Up to a sign factor, this local expression is exactly the same as the classical case. Moreover, one can show that the following expressions hold.

\begin{align*}
&\textnormal{Div}_{\p}(f\: X) = f \: \textnormal{Div}_{\p}X + (-1)^{\widetilde{f} \widetilde{X}} X(f);\\
&\textnormal{Div}_{\p'}X = \textnormal{Div}_{\p}X + X(f');\\
&\textnormal{Div}_{\p}[X,Y] = X(\textnormal{Div}_{\p} Y) -(-1)^{\widetilde{X} \widetilde{Y}}Y(\textnormal{Div}_{\p}X);
\end{align*}
where $X$ and $Y \in \Vect(M)$, $f \in C^{\infty}(M)$, and $\p' = \exp(f')\p$ with $f'\in C^{\infty}(M)$ is even. These properties, again up to some signs are identical to the properties of the classical divergence operator on a manifold. \par  
Much like the classical situation, a Riemannian metric defines  a canonical Berezin volume on $M$. This is well explained in \cite[Appendix B]{Voronov:2016} and our treatment of the construction is taken directly from there. The transformation rules for (components of) the metric can be written as 
\begin{align*}
g_{b'a'}(x') & = (-1)^{\widetilde{a}' \, \widetilde{b}} \left(\frac{\partial x^b}{\partial x^{b'}}\right) \left(\frac{\partial x^a}{\partial x^{a'}}\right)~  g_{ab}\\
 & = \left(\frac{\partial x^b}{\partial  x^{b'}}\right)g_{ba}\left(\frac{\partial x^a}{\partial  x^{a'}}\right)(-1)^{\widetilde{a}(\widetilde{a}' + 1)}.
\end{align*}
The third factor (along with the signs) is recognised as the supertranspose of the Jacobian matrix. Note that $\textnormal{Ber}(A^{st}) = \textnormal{Ber}(A)$. Thus, we obtain
\begin{align*}
\textnormal{Ber}(g_{b'a'}) & = \textnormal{Ber}\left(\frac{\partial x^b}{\partial x^{b'}} \right) ~ \textnormal{Ber}(g_{ba}) ~ \textnormal{Ber}\left(\frac{\partial x^a}{\partial x^{a'} }\right)\\
& = \textnormal{Ber}\left(\frac{\partial x^b}{\partial x^{b'}} \right)^2 ~~ \textnormal{Ber}(g_{ba}).
\end{align*}
Following classical notation, we set $|g| :=\textnormal{Ber}(g_{ba})$ and $ |g'| := \textnormal{Ber}(g_{b'a'})$,  and so we can write
$$|g'| = |g| ~~ \textnormal{Ber}\left( \frac{\partial x}{\partial x'}\right)^2.$$
\begin{definition}\label{def:ConBerVol}
Let $(M, g)$ be a Riemannian supermanifold. Then the \emph{canonical Berezin volume} is defined as
$$\rmd V :=  D[x]~  \sqrt{|g|},$$
where $|g| :=\textnormal{Ber}(g_{ba})$. 
\end{definition}
\begin{remark}
It should be noted that there is no canonical Berezin volume on an odd Riemannian supermanifold (or indeed, an odd symplectic supermanifold and this has important consequences for the Batalin--Vilkovisky formalism). The above considerations cannot be repeated for odd structures.
\end{remark}
In complete parallel with the classical case, the divergence of a vector field with respect to the canonical Berezin volume is related to the trace of the Lie derivative of the metric.
\begin{proposition}\label{prop:LieDiv}
Let $(M,g)$ be a Riemannian supermanifold and let $\rmd \textnormal{V}$ be the canonical Berezin volume. Then
$$\frac{1}{2} \textnormal{Str}_g\big(L_X g \big) = \textnormal{Div}_{\rmd \textnormal{V}}X\,.$$
\end{proposition}
\begin{proof}
Direct computation in local coordinates produces
$$(-1)^{\widetilde{a}} \frac{1}{2} g^{ab}(L_X g)_{ba} = (-1)^{\widetilde{a}(\widetilde{X}+1)} \left (\frac{\partial X^a}{\partial x^a} \right) + \frac{1}{2} \left(X^c \frac{\partial g_{ab}}{\partial x^c} \right)g^{ba}\,.$$
Next, we need the well-known formula $\delta \textnormal{Ber}(A) = \textnormal{Ber}(A) \, \textnormal{Str}(\delta A \, A^{-1} )$, which implies
$$\frac{1}{2}\textnormal{Str}(\delta A \, A^{-1}) = \frac{1}{\sqrt{\textnormal{Ber}(A)}} ~ \delta \sqrt{\textnormal{Ber}(A)}\,.$$
Thus,
\begin{align*}
 (-1)^{\widetilde{a}} \frac{1}{2} g^{ab}(L_X g)_{ba} & = (-1)^{\widetilde{a}(\widetilde{X} + 1)}\left(\frac{\partial X^a}{\partial x^a}\right)  + \frac{1}{\sqrt{|g|}} ~  X^a \frac{\partial  \sqrt{|g|}}{\partial x^a}\\
 & = (-1)^{\widetilde{a}(\widetilde{X} + 1)} ~ \frac{1}{\sqrt{|g|}}~\left(\frac{\partial  }{\partial x^a}( X^a \sqrt{|g|})\right)\,.
\end{align*}
Comparing this with \eqref{eqn:DivLC} (and using Definition \ref{def:ConBerVol}) establishes the proposition. 
\end{proof}
\begin{proposition}\label{prop:KilDivLes}
Let $(M,g)$ be a Riemannian supermanifold. If $X\in \Vect(M)$ is a Killing vector field then it is divergenceless (with respect to the canonical Berezin volume). 
\end{proposition}
\begin{proof}
This is a direct consequence of Proposition \ref{prop:LieDiv} together with Definition \ref{def:KillVF}.
\end{proof}
\begin{corollary}
The canonical Berezin volume on a Riemannian supermanifold $(M,g)$ is invariant under the action of a Killing vector field, i.e.,
$$L_X \rmd V =0\,,$$
if $X \in \Vect(M)$ is a Killing vector field.
\end{corollary}

\section{Q-manifolds and their modular classes}\label{sec:QMod}
\subsection{Homological vector fields and Q-manifolds}
We now turn our attention to homological vector fields and Q-manifolds.
\begin{definition}
A \emph{Q-manifold}  is a supermanifold  $M$, equipped with a distinguished  odd vector field $Q\in \Vect(M)$  that `squares to zero', i.e., $Q^{2} = \half [Q,Q] =0$. The vector field $Q$ is referred to as a \emph{homological vector field}  or a \emph{Q-structure}.
\end{definition}
Note that due to extra signs that appear in supergeometry, $[Q,Q] := Q\circ Q + Q \circ Q$, and hence  $Q^2 =0$ is a non-trivial  condition. In local coordinates, we have $Q = Q^{a}(x)\frac{\partial}{\partial x^{a}}$, and the condition that $Q$ is homological is 
$$Q^{2} = 0    \Longleftrightarrow     Q^{a}\frac{\partial Q^{b}}{\partial x^{a}} =0.$$
\begin{definition}
Let $(M_{1}, Q_{1})$ and $(M_{2}, Q_{2})$ be Q-manifolds. Then a morphism of supermanifolds $\psi : M_{1} \rightarrow M_{2}$ is a \emph{morphism of Q-manifolds} if it relates the two  homological vector fields, i.e.,
$$Q_{1} \circ \psi^{*} = \psi^{*} \circ Q_{2}.$$
\end{definition}
To be explicit, let us employ local coordinates $x^{a}$ on $M_{1}$ and $y^{\alpha}$ on $M_{2}$. We will write, using standard abuses of notation $\psi^{*}y^{\alpha} = \psi^{\alpha}(x)$. The statement that $\psi$ be a morphism of Q-manifolds means locally that
$$Q_{1}^{a}(x)\frac{\partial \psi^{\alpha}(x)}{\partial x^{a}} = Q_{2}^{\alpha}(\psi(x)).$$
Evidently, we obtain the category of Q-manifolds via standard composition of supermanifold morphisms. 
\begin{definition}
The \emph{standard cochain complex} associated with a Q-manifold is the  $\Z_{2}$-graded cochain complex $(C^{\infty}(M), Q)$. The resulting cohomology is referred to as the \emph{standard cohomology}  of the Q-manifold.
\end{definition}
We then see that morphisms of Q-manifolds are cochain maps between the respective standard cochain complexes.
\begin{theorem}[Shander \cite{Shander:1980}]\label{thrm:Shander}
Let $Q$ be a homological vector field on a superdomain $\mathcal{U}^{p|q}$, then the following are equivalent:
\begin{enumerate}
\item $Q$ is weakly non-degenerate at all points $p \in \mathcal{U}^p$, i.e., not all the components of $Q$ vanish at any given point; 
\item there exists a coordinate system $(x^1 , \cdots, x^p; \zx^1, \cdots, \zx^q)$ on $\mathcal{U}^{p|q}$ such that
$$Q = \frac{\partial}{\partial \zx^1}\,.$$ 
\end{enumerate} 
\end{theorem}
The above theorem tells us that locally and assuming that the homological vector field weakly non-degenerate on some appropriate neighbourhood,  then we can employ local coordinates $ x^a = (x^\mu , \zx^\lambda, \tau)$, where $\mu = 1, \cdots p$ and $\lambda = 1, \cdots, q-1$.  This theorem was extended by Vaintrob \cite{Vaintrob:1996} in the following way.
\begin{theorem}[Vaintrob \cite{Vaintrob:1996}]\label{thrm:Vaintrob}
Let $Q$ be a homological vector field on a supermanifold $M$. If $Q$ is non-singular (i.e., weakly non-degenerate in  neighbourhoods of any point on $|M|$), then there exists another a supermanifold $N$, such that
$$M \simeq N \times \R^{0|1}\,,$$ 
and the homological vector field takes the form 
$$Q = \frac{\partial}{\partial \tau}\,,$$
where $\tau$ is the global coordinate on $\R^{0|1}$. 
\end{theorem}

\subsection{Modular classes of Q-manifolds}
The modular class of a Q-manifold (\cite{Lyakhovich:2004,Lyakhovich:2010}) is defined in terms of the divergence (see \ref{eqn:DivVect}) of the homological vector field.
\begin{definition}
The \emph{modular class} of a Q-manifold  is the standard cohomology class of $\textnormal{Div}_{\p}Q$, i.e., 
$$\textnormal{Mod}(Q) := [\textnormal{Div}_{\p}Q]_{\textnormal{St}}.$$
\end{definition}
The modular class is independent of any chosen Berezin volume as any other choice of volume leads to divergences that differ only by something Q-exact, and so Q-closed (this follows directly from the properties of the divergence operator). This means that the modular class is a characteristic class of a Q-manifold.  The vanishing of the modular class is a necessary and sufficient condition for the existence of a Berezin volume that is Q-invariant. \par 
In some given set of local coordinates, one can write out the divergence as
$$\textnormal{Div}_{\p}Q =  \frac{\partial Q^{a}}{\partial x^{a}} + Q(\log(\rho)).$$
The \emph{local (characteristic) representative} of the modular class is understood as just the term
\begin{equation}\label{eqn:localrep}
\phi_{Q}(x) := \frac{\partial Q^{a}}{\partial x^{a}}(x).
\end{equation}
In general, this term is \emph{not} invariant under changes of coordinates, only the full expression for the divergence is.  However, as we are always dropping terms that are Q-exact, the local representative is still meaningful, though as written it is only a local function on $M$.
\begin{remark}
The expression (\ref{eqn:localrep}) gives the local representative of the standard (coordinate) volume. In general we do not have a version of the Poincar\'{e} lemma: meaning that Q-closed functions are \emph{not} necessarily locally Q-exact. Thus, it makes sense to speak of a local representative of the modular class.
\end{remark}
\begin{definition}\label{def:UniMod}
A Q-manifold $(M,Q)$ is said to be a \emph{unimodular Q-manifold} if its modular class vanishes. In other words, if there exists a Q-invariant Berezinian volume.
\end{definition}
\begin{example}
The prototypical example of a Q-manifold is the antitangent bundle $\Pi \sT M$. In natural local coordinates $(x^a, \rmd x^b)$, we have the de Rham differential
$$\rmd = \rmd x^a \frac{\partial}{\partial x^a}\,.$$
Clearly, the local representative of the modular class vanishes and so $\Pi \sT M$ is unimodular. The invariant Berezin volume is just the canonical coordinate volume $D[x, \rmd x]$.
\end{example}

\section{Riemannian Q-manifolds}\label{sec:RiemQ}
\subsection{Homological-Killing vector fields}
If a supermanifold is both simultaneously  a Riemannian supermanifold and a Q-manifold, we have the natural question of the compatibility of the two structures. In practice, this often reduces to one structure generating a symmetry of the other and maybe vice-versa. We, therefore, make the following definition.
\begin{definition}
Let $(M,g)$ be a Riemannian supermanifold. Then a \emph{homological-Killing vector field} $Q \in \Vect(M)$ is a homological vector field that is also a Killing vector field. That is, it satisfies
\begin{enumerate}
\itemsep0.5em
\item $Q^2 = \frac{1}{2}[Q,Q] =0$, and,
\item $L_Qg =0$.
\end{enumerate}
\end{definition}
\begin{remark}
The standard cohomology of a Q-manifold can be extended to all tensor fields on $(M,Q)$ via the Lie derivative. In particular, $(\mathcal{A}^2(\sT M), L_Q)$ is a $\Z_2$-graded cochain complex. Thus, the Killing condition of a homological vector field can be restated as the metric $g$ being $Q$-closed.
\end{remark}
\begin{definition}
A \emph{Riemannian Q-manifold} is a triple $(M, g, Q)$, where $(M,g)$ is a Riemannian manifold, $(M,Q)$ is a Q-manifold such that $Q$ is a homological-Killing vector field.
\end{definition}

\begin{example}[Euclidean superspace] Consider $\R^{1|2}$ equipped with global coordinates $(t, \zx^1, \zx^2)$ and with standard metric
$$g =  \dot{t}^2 \pm 2 \, \dot{\zx}^1 \dot{\zx}^2.$$
This metric is clearly invariant under translations of any of the even or odd directions. We may take 
$$Q = \frac{\partial}{\partial \zx^1}$$
as our distinguished homological-Killing vector field in this particular chart.
\end{example}
\begin{example}[Positive half-superline] Consider $\R^{1|2}$ equipped with global coordinates $(t, \zx^1, \zx^2)$. The positive half-superline  $\R^{1|2}_{>0}$ we define to be the open subsupermanifold of $\R^{1|2}$ defined by $t >0$. We equip the positive half-superline with the metric
$$g =  (\dot{t}^2 \pm 2 \, \dot{\zx}^1 \dot{\zx}^2) ~t^{-2}.$$
This metric is clearly invariant under translation in either of the odd directions. However, unlike the previous example, it is not invariant under translations in the even direction.  We may take 
$$Q = \frac{\partial}{\partial \zx^1}$$
as our distinguished homological-Killing vector field in this particular chart.
\end{example}
We will shortly see that the above examples are somewhat generic (see Proposition \ref{prop:LocForm} and Corollary \ref{cor:LocForm}).
\begin{definition}
A \emph{morphism between two Riemannian Q-manifolds}
$$\phi :  (M, g, Q) \longrightarrow  (m', g', Q'),$$
is a morphism of supermanifolds such that
\begin{enumerate}
\itemsep0.5em
\item $\phi^* g' = g$, and,
\item $Q \circ \phi^* = \phi^* \circ Q'.$
\end{enumerate}
\end{definition}
 In local coordinates the two above condition can be written in the following way. If we consider local coordinates $x^a$ on $M$ and $y^\alpha$ on $M'$, and then denote $\phi^*y^\alpha = \phi^\alpha(x)$, then  we can write
 \begin{align*}
&  (-1)^{\widetilde{a} \, \widetilde{\alpha}} \left(\frac{\partial \phi^\beta(x)}{\partial x^b} \right) \left(\frac{\partial \phi^\alpha(x)}{\partial x^a} \right)g_{\alpha\beta}(\phi(x)) = g_{ab}(x),\\
& Q^a(x)\left(\frac{\partial \phi^\alpha(x)}{\partial x^a} \right) = Q^\alpha(\phi(x)).
\end{align*}
  One can quickly see that morphisms between Riemannian Q-manifolds can be composed (as morphisms between supermanifolds) and that in this way we obtain the category of Riemannian Q-manifolds.

\subsection{The modular class of a Riemannian Q-manifold}
We are now in a position to state the following.
\begin{theorem}
Let $(M, g , Q)$ be a Riemannian Q-manifold. Then as a Q-manifold, $(M,Q)$ is unimodular (see Definition \ref{def:UniMod}).
\end{theorem}
\begin{proof}
From Proposition \ref{prop:KilDivLes} we see that any Killing vector field has vanishing divergence with respect to the canonical Berezin volume.   From the definition of the divergence,  it is clear that $L_Q \rmd \mathrm{V} =0$. Moreover, the existence of a $Q$-invariant Berezin volume  is equivalent to the vanishing of the modular class. Hence, the Q-manifold $(M,Q)$ is unimodular. 
\end{proof}
\begin{remark}
It is clear that not all unimodular Q-manifolds can be equipped with a Riemannian metric that renders them a Riemannian Q-manifold. For one, we require the dimension of the supermanifold to be $n|2 \,p$. This immediately rules out the possibility of constructing a Riemannian metric on $\Pi \sT M$ such that the de Rham differential $\rmd$ is a Killing vector field. However, it is known that odd Riemannian metrics exist for which the de Rham differential is Killing. See Monterde \& S\'{a}nchez-Valenzuela \cite{Monterde:1996} for details. 
\end{remark}
\begin{example}
Let $(\mathfrak{g}, [-,-]$ be a (non-super) Lie algebra of dimension $2p$. Furthermore, let us assume that this Lie algebra comes equipped with an almost symplectic structure, i.e., a Lie algebra two form of maximal rank, which is not necessarily closed with respect to the Chevalley--Eilenberg differential. Let us now pass to the ``super-picture''.  As standard, $\Pi \mathfrak{g}$ is a Q-manifold, were, in natural linear coordinates, the homological vector field is
$$Q = \frac{1}{2}\zx^\alpha \zx^\beta Q_{\beta \alpha}^\gamma \frac{\partial}{\partial \zx^\gamma},$$
here $Q_{\beta \alpha}^\gamma$ are the structure constants of the Lie algebra.  The Jacobi identity for the Lie bracket is equivalent to $Q^2 =0$. The almost symplectic structure we can interpret as a Riemannian metric on $\Pi \mathfrak{g}$,
$$g = \dot{\zx}^\alpha \dot{\zx}^\beta g_{\beta \alpha}\,,$$
where $g_{\beta \alpha} = - g_{\alpha \beta}$. The Killing equation reduces to the algebraic condition
\begin{equation}\label{eqn:LieKillEq}
Q^\gamma_{\delta \alpha}g_{\gamma \beta} - Q^\gamma_{\delta \beta}g_{\gamma \alpha} =0.
\end{equation}
If \eqref{eqn:LieKillEq} holds, then $(\Pi \mathfrak{g},Q)$ is a Riemannian Q-manifold.  Assuming that this is the case, then $\mathfrak{g}$ is a unimodular Lie algebra in the classical sense. Note that the Killing equation is a more restrictive condition that just unimodularity of the Lie algebra.  To see the classical unimodularity, consider contraction of the Killing equation
\begin{equation*}
Q^\gamma_{\delta \alpha}g_{\gamma \beta}g^{\beta \epsilon} - Q^\gamma_{\delta \beta}g_{\gamma \alpha}g^{\beta \epsilon}  = Q^\epsilon_{\delta \alpha}- Q^\gamma_{\delta \beta}g_{\gamma \alpha}g^{\beta \epsilon} =0.
\end{equation*}
Now setting $ \epsilon = \alpha$, as this is what we are interested in when it comes to unimodularity, gives
\begin{equation*}
Q^\alpha_{\delta \alpha} -  Q^\gamma_{\delta \beta}g_{\gamma \alpha}g^{\beta \alpha}  = Q^\alpha_{\delta \alpha} + Q^\gamma_{\delta \beta}g^{\beta \alpha}g_{\alpha \gamma}=  Q^\alpha_{\delta \alpha}+ Q^\beta_{\delta \beta} =0.
\end{equation*}
Thus, $Q^{\alpha}_{\beta \alpha} =0$, which is precisely the condition that $\mathfrak{g}$ be unimodular.
\end{example}

\subsection{Killing--Shander coordinates}
 Assuming that $Q$ is weakly non-degenerate in the neighbourhood of a point $p \in |U|$, then using Theorem \ref{thrm:Shander}, we can employ local coordinates $x^a = (x^i , \tau)$. In these privileged coordinates, the Killing equation  (see Definition \ref{def:KillVF} and \eqref{eqnm:LieDerMet}) reduces to 
 \begin{equation}\label{eqn:RedKillEQ}
 \frac{\partial g_{ba}}{\partial \tau} = 0.
 \end{equation}
 We will refer to this choice of coordinates as \emph{Killing--Shander coordinates}. Thus we are lead to the following:
 \begin{proposition}\label{prop:LocForm}
 Let $(M, g,Q)$ be a Riemannian Q-manifold.  In the neighbourhood of a point $p \in |U| \subset |M|$ on which $Q$ is weakly non-degenerate, there exists coordinates $x^a := (x^i , \tau)$ such that all the components of the Riemannian metric are independent of $\tau$. Conversely, if in the neighbourhood of any point on $|M|$ there exists coordinates $x^a := (x^i , \tau)$ such that all the components of the Riemannian metric are independent of $\tau$, then there exists a nowhere vanishing homological vector field $Q$.
 \end{proposition}
 \begin{proof}
 The first part of the proposition is a direct consequence of \eqref{eqn:RedKillEQ}.  The converse statement follows as in the  given coordinate systems  $Q = \frac{\partial}{\partial \tau}$ is clearly homological and Killing. The homological Killing vector field must be nowhere vanishing in order for the required coordinates to exists in the neighbourhood of any point. 
 \end{proof}
 \begin{corollary}\label{cor:LocForm}
 With the conditions of the previous proposition in place,  in Killing--Shander coordinates the metric  has the form
 $$g = \dot{x}^i \dot{x}^j g_{ji}(x)  +  \dot{\tau} \dot{x}^i g_i(x)\, ,$$
 and the homological Killing vector field has the form
 $$ Q = \frac{\partial}{\partial \tau}\,. $$
 \end{corollary}
Using Theorem \ref{thrm:Vaintrob}, if we have a nowhere vanishing homological vector field, then we can consider $M \simeq N \times \R^{0|1}$ as a trivial odd line bundle. Thus, changes of Killing--Shander coordinates are of the form 
\begin{align*}
x^{i'} = x^{i'}(x), && \tau' = c \, \tau,
\end{align*}
where $c \in \R_{*}$. The naturally induced changes of coordinates on the tangent bundle are 
\begin{align*}
\dot{x}^{i'} =\dot{x}^i \left(\frac{ \partial x^{i'}}{\partial x^i}\right), && \dot{\tau}' = c \, \dot{\tau}.
\end{align*}
Then, examining the local form of the metric show that the term $\dot{x}^i \dot{x}^j g_{ji}(x)$ belongs to $\mathcal{A}^2(\sT N)$. However, it is \emph{not} a Riemannian metric as $N$ has an odd number of `odd directions', i.e., locally we have an odd number of anticommuting coordinates. Examining the second term, we see that we have the transformation rule
$$g_{i'} = c^{-1}\left( \frac{\partial x^i}{\partial x^{i'}}\right)g_i\,,$$
and this  term as the interpretation (under the specified coordinate changes) as an odd twisted covariant one-form. Under these transformations, the homological vector field transforms by an  irrelevant rescaling by $c^{-1}$, i.e., simply rescaling the odd coordinates again will remove this factor.

\section{Concluding remarks}\label{sec:Conc}
We have shown, rather explicitly,  that Riemannian Q-manifolds represent a large class of unimodular Q-manifolds, i.e., supermanifolds that admit a Q-invariant Berezin volume.  The  Q-invariant volume is just the canonical Berezin volume associated with the (even) Riemannian metric.   If instead of an even metric one considers an odd metric, then a Berezin volume needs to be separately specified. Thus, in general, a Killing vector field on an odd Riemannian supermanifold does not automatically preserve the volume. This is, of course, in complete parallel with the case of even and odd symplectic supermanifolds and Hamiltonian vector fields.\par 
It would be interesting to construct further examples of Riemannian Q-manifolds and examine the interplay between their standard cohomology and their Riemannian geometry.  To the author's knowledge there has been no published works in this direction.

\section*{Acknowledgements}
The author cordially thanks Steven Duplij for his helpful comments on an earlier draft of this work.

\end{document}